\RequirePackage{amsmath}
\documentclass[runningheads]{llncs}
\usepackage[T1]{fontenc}

\usepackage{amssymb, dsfont, physics, textcomp}
\usepackage{quantikz}
\usepackage{graphicx}
\usepackage{orcidlink}
\usepackage{hyperref}
\usepackage[capitalize,nameinlink,sort,noabbrev]{cleveref}

\newcommand{\s}[1]{\{#1\}}
\newcommand{\calu}{\mathcal{U}}
\newcommand{\calv}{\mathcal{V}}
\newcommand{\calw}{\mathcal{W}}
\newcommand{\D}{\mathbb{D}}
\newcommand{\Z}{\mathbb{Z}}
\newcommand{\C}{\mathbb{C}}
\newcommand{\Q}{\mathbb{Q}}
\newcommand{\zz}[1]{\zeta_{2^{#1}}}
\newcommand{\zzk}{\zeta_{2^k}}
\newcommand{\Zz}[1]{\mathbb{Z}[\zeta_{2^{#1}}]}
\newcommand{\Zzk}{\mathbb{Z}[\zeta_{2^k}]}
\newcommand{\Dz}[1]{\mathbb{Z}[1/2,\zeta_{2^{#1}}]}
\newcommand{\Dzk}{\mathbb{Z}[1/2,\zeta_{2^k}]}
\newcommand{\DDz}[1]{\mathbb{D}[\zeta_{2^{#1}}]}
\newcommand{\DDzk}{\mathbb{D}[\zeta_{2^k}]}
\newcommand{\uz}[1]{\mathrm{U}(\Dz{#1})}
\newcommand{\uzk}{\mathrm{U}(\Dzk)}
\newcommand{\unnk}[1]{\mathrm{U}_{#1}(\Dzk)}

\newcommand{\DDuz}[1]{\mathrm{U}(\DDz{#1})}
\newcommand{\DDuzk}{\mathrm{U}(\DDzk)}
\newcommand{\DDunnk}[1]{\mathrm{U}_{#1}(\DDzk)}
\newcommand{\DDunnkk}[2]{\mathrm{U}_{#1}(\DDz{#2})}
\newcommand{\gensk}{\mathcal{G}_{2^k}}
\newcommand{\gensn}[1]{\mathcal{G}_{#1}}
\newcommand{\hgate}{H'}

\usepackage{xcolor}
\hypersetup{
    colorlinks,
    linkcolor={red!80!black},
    citecolor={green!80!black},
    urlcolor={blue!80!black}
}

\begin{document}

\title{
  Exact Synthesis of Multiqubit \texorpdfstring{\\}{} 
  Clifford-Cyclotomic Circuits
}

\author{
  Matthew Amy\inst{1}\orcidlink{0000-0003-3514-420X}\and
  Andrew N. Glaudell\inst{2}\orcidlink{0000-0001-9824-5804}\and
  Shaun Kelso\inst{3}\orcidlink{0009-0007-1420-5938}\and \\
  William Maxwell\inst{3,4}\orcidlink{0009-0005-8603-2955}\and
  Samuel S. Mendelson\inst{3}\orcidlink{0000-0002-8413-9872}\and
  Neil J. Ross\inst{5}\orcidlink{0000-0003-0941-4333}
}

\authorrunning{M. Amy et al.}

\institute{Simon Fraser University, Burnaby, BC, Canada\\
\email{matt\_amy@sfu.ca}\and
Photonic Inc., Vancouver, BC, Canada\\
\email{andrewglaudell@gmail.com}\and
NSWC Dahlgren Division, Dahlgren, VA, U.S.A.\\
\texttt{shaun.f.kelso.civ@us.navy.mil, samuel.mendelson@gmail.com}\and
Sandia National Laboratories, Albuquerque, NM, U.S.A.\\
\email{wjmaxwell@sandia.gov}\and
Dalhousie University, Halifax, NS, Canada\\
\email{neil.jr.ross@dal.ca}
}  

\maketitle

\begin{abstract}
Let $n\geq 8$ be divisible by 4. The Clifford-cyclotomic gate set $\mathcal{G}_n$ is the universal gate set obtained by extending the Clifford gates with the $z$-rotation $T_n = \mathrm{diag}(1,\zeta_n)$, where $\zeta_n$ is a primitive $n$-th root of unity. In this note, we show that, when $n$ is a power of 2, a multiqubit unitary matrix $U$ can be exactly represented by a circuit over $\mathcal{G}_n$ if and only if the entries of $U$ belong to the ring $\Z[1/2,\zeta_n]$. We moreover show that $\log(n)-2$ ancillas are always sufficient to construct a circuit for $U$. Our results generalize prior work to an infinite family of gate sets and show that the limitations that apply to single-qubit unitaries, for which the correspondence between Clifford-cyclotomic operators and matrices over $\Z[1/2,\zeta_n]$ fails for all but finitely many values of $n$, can be overcome through the use of ancillas.

\end{abstract}

\section{Introduction}
\label{sec:intro}

\subsection{Background}
\label{ssec:background}

Let $n\geq 8$ be an integer divisible by 4. The \textbf{single-qubit Clifford-cyclotomic gate set of degree $n$} was introduced in \cite{forest2015exact} and consists of the gates 
\[
\hgate = 
\frac{1}{2}
\begin{bmatrix}
1+i & 1+i \\
1+i & -1-i
\end{bmatrix}
\qquad
\mbox{and}
\qquad
T_{n} = 
\begin{bmatrix}
1     & \cdot \\
\cdot & \zeta_n
\end{bmatrix},
\]
where $\zeta_n=e^{2\pi i /n}$ is a primitive $n$-th root of unity, $\hgate=\zeta_8 H$ is equal to the usual \textbf{Hadamard gate} $H$ up to a global phase of $\zeta_8$, and $T_n$ is a $z$-rotation gate of order $n$. The gate $S=T_n^{n/4}$ is the usual \textbf{phase gate} and the gate $T_8$ is simply known as the \textbf{$T$ gate}. The single-qubit Clifford-cyclotomic gate set is a universal extension of the \textbf{single-qubit Clifford  gate set} $\s{\hgate,S}$; it coincides with the well-studied \textbf{single-qubit Clifford+$T$ gate set} when $n=8$.

The entries of $\hgate$ and $T_n$ lie in $\Z[1/2, \zeta_n]$, the smallest subring of $\C$ containing $1/2$ and $\zeta_n$. As a consequence, if a $2$-dimensional  unitary matrix $U$ can be exactly represented by a single-qubit Clifford-cyclotomic circuit of degree $n$, then the entries of $U$ belong to $\Z[1/2, \zeta_n]$. In their seminal 2012 paper \cite{kliuchnikov2013fast}, Kliuchnikov, Maslov, and Mosca proved that the converse implication holds when $n=8$: every $2$-dimensional unitary matrix with entries in $\Z[1/2,\zeta_8]$ can be exactly represented by a Clifford+$T$ circuit. Thus, single-qubit Clifford+$T$ operators correspond precisely to elements of $\mathrm{U}_2(\Z[1/2,\zeta_8])$, the group of $2\times 2$ unitary matrices over $\Z[1/2,\zeta_8]$. Forest et al. later showed in \cite{forest2015exact} that such a correspondence holds when $n$ is one of $8$, $12$, $16$, or $24$, but, disappointingly, that it fails for almost all other values of $n$. Ingalls et al. put the nail in this coffin in 2019 by proving that  $8$, $12$, $16$, and $24$ are in fact the only values of $n$ for which such a correspondence holds \cite{ingalls_jordan_keeton_logan_zaytman_2021}, as had been previously conjectured by Sarnak~\cite{sarnak}.

The \textbf{multiqubit Clifford-cyclotomic gate set of degree $n$}, which we denote $\mathcal{G}_n$, is obtained by adding the \textbf{controlled-NOT gate} 
\[
CX= I_2 \oplus \begin{bmatrix}
\cdot & 1 \\
1 & \cdot
\end{bmatrix}
\]
to the single-qubit Clifford-cyclotomic gate set of degree $n$. In other words, $\mathcal{G}_n$ is the extension of the \textbf{multiqubit Clifford gate set} $\s{\hgate, S, CX}$ by the $z$-rotation $T_n$. For convenience, we set $\mathcal{G}_2 = \s{X, CX, CCX, H\otimes H}$ and $\mathcal{G}_4 = \s{X, CX, CCX, S, \hgate}$, where 
\[
X= \begin{bmatrix}
\cdot & 1 \\
1 & \cdot
\end{bmatrix},
\qquad
CCX = I_6\oplus X,
\qquad
\mbox{and}
\qquad
H\otimes H = \frac{1}{2}
\begin{bmatrix}
1 &  1 &  1 &  1 \\
1 & -1 &  1 & -1 \\
1 &  1 & -1 & -1 \\
1 & -1 & -1 &  1
\end{bmatrix}.
\]
The gates $X$ and $CCX$ are the usual \textbf{NOT gate} and \textbf{doubly-controlled-NOT gate} (or \textbf{Toffoli gate}), respectively.

In \cite{Giles2013a}, Giles and Selinger extended Kliuchnikov, Maslov, and Mosca's 2012 result to the multiqubit setting by proving that a unitary matrix $U$ of dimension $2^m$ can be represented by an $m$-qubit circuit over $\mathcal{G}_8$ if and only if the entries of $U$ lie in the ring $\Z[1/2, \zeta_8]$. In \cite{Amy2020}, some of the present authors showed how to adapt the methods of Giles and Selinger to a handful of other gate sets, including $\mathcal{G}_2$ and $\mathcal{G}_4$. In the multiqubit context, circuits can use ancillary qubits, provided that they are initialized and terminated in the computational basis state $\ket{0}$. It was shown in \cite{Amy2020} and \cite{Giles2013a} that a single ancilla is always sufficient to construct the desired circuits. 

Clifford-cyclotomic circuits, and in particular those of degree $2^k$ for some positive integer $k$, are  ubiquitous in quantum computation; they appear in Shor's factoring algorithm \cite{shor}, the study of the Clifford hierarchy \cite{hierarchy}, and protocols for state distillation \cite{distillation}.

\subsection{Contributions}
\label{ssec:contribs}

Let $k$ and $m$ be positive integers. In the present note, we show that a $2^m$-dimensional unitary matrix $U$ can be exactly represented by an $m$-qubit Clifford-cyclotomic circuit of degree $2^k$ if and only if the entries of $U$ lie in the ring $\Dzk$. To construct a circuit for $U$, a single ancilla suffices, when $k\leq 2$, and $k-2$ ancillas suffice, when $k>2$.

Our results extend those of \cite{Amy2020} and \cite{Giles2013a} to an infinite family of multiqubit gate sets, but our proof is surprisingly simple. It relies on the fact that the root of unity $\zzk$ can be represented by a 2-dimensional unitary matrix over $\Dz{k-1}$, and that this representation can be used to define a well-behaved function $\phi_k:\uzk \to \uz{k-1}$ mapping unitary matrices over $\Dzk$ to unitary matrices over $\Dz{k-1}$. The function $\phi_k$ generalizes the standard real representation of complex numbers which was used by Aharonov in \cite{aharonov2003} to prove the universality of the Toffoli-Hadamard gate set and is an example of a \textbf{catalytic embedding} \cite{catemb}. One can think of our results as circumventing the no-go theorems of \cite{forest2015exact} and \cite{ingalls_jordan_keeton_logan_zaytman_2021} through the use of ancillas: there are elements of $\unnk{2}$ that cannot be represented by an ancilla-free single-qubit circuit over $\gensk$, but every such element becomes representable if sufficiently many additional qubits are available.

\subsection{Contents}
\label{ssec:contents}

The note is organized as follows. In \cref{sec:cyclo}, we briefly review some important properties of the ring $\Dzk$. We introduce catalytic embeddings in \cref{sec:cats} and define the catalytic embedding $\phi_k$. \cref{sec:exactsynth} contains the proof of our main result. We discuss future work in \cref{sec:conc}.

\section{Cyclotomic Integers}
\label{sec:cyclo}

We start by briefly discussing the rings of \textbf{cyclotomic integers} that will be of interest in the rest of the note. For further details, the reader is encouraged to consult \cite{cyclo}.

Let $k$ be a positive integer. The ring $\Zzk$ is the smallest subring of $\C$ containing $\zzk$. Hence, $\Zz{1}=\Z$. Moreover, when $k>1$, we have $\zzk^2=\zz{k-1}$ and therefore $\Zz{k-1}\subseteq \Zzk$. It will be useful for our purposes to further note that, for $k>1$, 
\begin{equation}
\label{eq:zzk}
\Zzk = \s{a+b\zz{k}\mid a,b\in\Zz{k-1}}.
\end{equation}
The linear combinations in \cref{eq:zzk} are unique. That is, every element of $\Zzk$ can be uniquely written as $a+b\zz{k}$, for some $a,b\in\Zz{k-1}$.

We will be interested in an extension of $\Zzk$ obtained by localizing $\Zzk$ at 2, i.e., by adding denominators that are powers of 2. The resulting ring is 
\begin{equation}
\label{eq:dzk}
\Dzk = \s{a/ 2^\ell \mid a\in \Zzk, \ell \in \Z}.
\end{equation}

For brevity, and in keeping with prior work (see, e.g., \cite{Amy2020,Giles2013a}), we denote $\Dzk$ by $\DDzk$ in what follows. This notation emphasizes the fact that $\Dzk$ can be seen as the extension by $\zzk$ of the ring $\D=\s{a/2^\ell \mid a\in \Z, \ell \in \Z }$ of \textbf{dyadic rationals}.

\begin{lemma}
\label{lem:uniqueness}
Let $k\geq 2$. Every element of $\DDzk$ can be uniquely written as $a+b\zzk$ for some $a,b\in\DDz{k-1}$.
\end{lemma}

\begin{proof}
\cref{eq:zzk,eq:dzk} jointly imply that every element of $\DDzk$ can be written as $a+b\zzk$ for some $a,b\in\DDz{k-1}$. To see that this expression is unique, let $a,b,a',b'\in \DDz{k-1}$ and suppose that $a+b\zzk = a'+b'\zzk$. By choosing $\ell$ large enough, $2^\ell(a+b\zzk) = 2^\ell(a'+b'\zzk)$ becomes an equation over $\Zzk$, from which we get $a=a'$ and $b=b'$.
\end{proof}

\section{Catalytic Embeddings}
\label{sec:cats}

We now define \textbf{catalytic embeddings}. The definition introduced below is a special case of the more general notion of catalytic embedding used in \cite{catemb}, but it suffices for our purposes.

Let $\calu$ and $\calv$ be collections of unitaries. An \textbf{$\ell$-dimensional catalytic embedding} of $\calu$ into $\calv$ is a pair $(\phi,\ket{\lambda})$ consisting of a function $\phi:\calu \to \calv$ and a quantum state $\ket{\lambda}\in\C^\ell$ such that if $U\in\calu$ has dimension $d$ then $\phi(U)\in\calv$ has dimension $d\ell$, and
\begin{equation}
\label{eq:catalytic}
\phi(U) (\ket{u}\otimes\ket{\lambda}) = (U\ket{u})\otimes \ket{\lambda}
\end{equation}
for every $\ket{u}\in\C^d$. We refer to the state $\ket{\lambda}$ as the \textbf{catalyst} and to \cref{eq:catalytic} as the \textbf{catalytic condition}. We sometimes write  $(\phi,\ket{\lambda}):\calu\to\calv$ to indicate that $(\phi,\ket{\lambda})$ is a catalytic embedding of $\calu$ into $\calv$. If $(\phi,\ket{\lambda}):\calu\to\calv$ and $(\psi,\ket{\omega}):\calv\to\calw$ are catalytic embeddings, then $(\psi\circ \phi, \ket{\lambda}\otimes\ket{\omega})$ is a catalytic embedding of $\calu$ into $\calw$, since
\[
\psi(\phi(U))(\ket{u}\otimes\ket{\lambda}\otimes\ket{\omega}) = (\phi(U)(\ket{u}\otimes\ket{\lambda}))\otimes\ket{\omega} = (U\ket{u})\otimes\ket{\lambda}\otimes\ket{\omega}.
\]
We refer to this catalytic embedding as the \textbf{concatenation} of $(\phi,\ket{\lambda})$ and $(\psi,\ket{\omega})$. The concatenation of catalytic embeddings is associative and $(I_\calu, [1]):\calu \to \calu$ acts as the identity for concatenation.

Now let $\DDuzk$ denote the collection of all unitary matrices over $\DDzk$. The rest of this section is dedicated to constructing, for every $k\geq 2$, a $2$-dimensional catalytic embedding $\DDuzk \to \DDuz{k-1}$. To this end, we define the state $\ket{\lambda_k}$ and the matrix $\Lambda_k$ as
\[
\ket{\lambda_k} = \frac{1}{\sqrt 2}\begin{bmatrix} 1 \\ \zzk \end{bmatrix}
\quad
\mbox{ and }
\quad
\Lambda_k = 
\begin{bmatrix}
0 &  1 \\
\zz{k-1} & 0 
\end{bmatrix},
\]
respectively. Note that $\Lambda_k$ is a unitary matrix and $\ket{\lambda_k}$ is an eigenvector of $\Lambda_k$ for eigenvalue $\zzk$. To verify the latter claim, we compute:
\begin{equation}
\label{eq:eigen}
\Lambda_k \ket{\lambda_k} = \begin{bmatrix}
0 &  1 \\
\zz{k-1} & 0 
\end{bmatrix} \frac{1}{\sqrt 2}\begin{bmatrix} 1 \\ \zzk \end{bmatrix}
=\frac{1}{\sqrt 2}\begin{bmatrix} \zzk \\ \zz{k-1} \end{bmatrix}
=\frac{1}{\sqrt 2}\begin{bmatrix} \zzk \\ \zzk^2 \end{bmatrix}
=\zzk\ket{\lambda_k}.
\end{equation}
Note further that $\zzk^\dagger = \zz{k-1}^\dagger\zzk$ and that $\Lambda_k^\dagger = \zz{k-1}^\dagger \Lambda_k$. In order to define the desired catalytic embedding, we start by showing that the matrix $\Lambda_k$ can be used to define a function $\DDuzk \to \DDuz{k-1}$.

\begin{lemma}
\label{lem:unitarity}
Let $k\geq 2$, let $A$ and $B$ be matrices over $\DDz{k-1}$, and assume that $A+B\zzk\in\DDuzk$. Then $A\otimes I + B\otimes \Lambda_k\in\DDuz{k-1}$.
\end{lemma}

\begin{proof}
Let $k$, $A$, and $B$ be as stated. Since $A+B\zzk$ is unitary and $\zzk^\dagger = \zz{k-1}^\dagger\zzk$, we have 
\[
I = (A+B\zzk)^\dagger(A+B\zzk) 
  = (A^\dagger A + B^\dagger B) + (A^\dagger B +  B^\dagger A \zz{k-1}^\dagger)\zzk. 
\]
Hence, $A^\dagger A + B^\dagger B=I$ and $A^\dagger B +  B^\dagger A \zz{k-1}^\dagger =0$. Now, since $\Lambda_k^\dagger = \zeta_{2^{k-1}}^\dagger \Lambda_{k}$ and $\Lambda_k$ is unitary, we have 
\[
(A\otimes I + B\otimes \Lambda_k)^\dagger (A\otimes I + B\otimes \Lambda_k) 
 = (A^\dagger A + B^\dagger B)\otimes I + (A^\dagger B + B^\dagger A \zz{k-1}^\dagger)\otimes \Lambda_k 
 = I.
\]
Reasoning analogously, one can also show that $(A\otimes I + B\otimes \Lambda_k)(A\otimes I + B\otimes \Lambda_k)^\dagger =I$. This proves that $A\otimes I + B\otimes \Lambda_k$ is indeed unitary.
\end{proof}

\begin{proposition}
\label{prop:emb-step}
Let $k\geq 2$ and let $\phi_k:\DDuzk \to \DDuz{k-1}$ be the function defined by 
\[
\phi_k : A+B\zzk \mapsto A\otimes I + B\otimes \Lambda_k.
\]
Then the pair $(\phi_k,\ket{\lambda_k})$ is a 2-dimensional catalytic embedding of $\DDuz{k}$ into $\DDuz{k-1}$.
\end{proposition}

\begin{proof}
Every element $U$ of $\DDuzk$ can be uniquely written as $U=A + B\zzk$, where $A$ and $B$ are matrices over $\DDz{k-1}$. Hence, \cref{lem:unitarity} implies that $\phi_k:\DDuzk \to \DDuz{k-1}$ is indeed a function. Moreover, by construction, $\phi_k(U)\in \DDuz{k-1}$ has dimension $2d$, if $U\in\DDuzk$ has dimension $d$. Now let $\ket{u}\in\C^n$. Then 
\begin{align*}
\phi_k(U) (\ket{u}\otimes\ket{\lambda_k}) &= (A\otimes I + B \otimes \Lambda_{k}) (\ket{u}\otimes\ket{\lambda_k})\\
&= A\ket{u}\otimes I \ket{\lambda_k}+ B\ket{u} \otimes \Lambda_{k}\ket{\lambda_k} \\
&= A\ket{u}\otimes \ket{\lambda_k}+ B\ket{u} \otimes \zzk\ket{\lambda_k} \\ 
&= A\ket{u}\otimes \ket{\lambda_k}+ B\zzk \ket{u} \otimes \ket{\lambda_k} \\ 
&= (A\ket{u}+ B\zzk \ket{u})\otimes\ket{\lambda_k} \\ 
&=(U\ket{u})\otimes \ket{\lambda_k}. 
\end{align*}
Thus, $(\phi_k,\ket{\lambda_k})$ is a 2-dimensional catalytic embedding $\DDuzk \to \DDuz{k-1}$.
\end{proof}

\begin{remark}
\label{rem:standard}
The catalytic embedding of \cref{prop:emb-step} is an example of what is called a \textbf{standard catalytic embedding} in \cite{catemb}. At the heart of this construction lies the fact that $\zzk$ can be represented by the matrix $\Lambda_k$, whose characteristic polynomial is also the minimal polynomial of $\zeta_k$ over $\Q[\zz{k-1}]$. A more general description of this construction can be found in \cite{catemb}.
\end{remark}

\section{Exact Synthesis}
\label{sec:exactsynth}

We now prove our main result. While it is clear that if a unitary $U$ can be represented by a circuit over $\gensk$ then it is an element of $\DDuz{k}$, the challenge is to show that the converse implication is also true. The main idea behind the proof is to use \cref{prop:emb-step} to inductively reduce the problem for $\DDuz{k}$ to the problem for $\DDuz{k-1}$, and so on until one reaches a case for which the result is known, such as $\DDuz{3}$, $\DDuz{2}$, or $\DDuz{1}$. We formalize this intuition in the proposition below.

\begin{theorem}
\label{thm:exactsynth}
Let $k$ and $m$ be positive integers. A $2^m\times 2^m$ unitary matrix $U$ can be exactly represented by an $m$-qubit circuit over $\gensk$ if and only if  $U\in \DDunnk{2^m}$. Moreover, to construct a circuit for $U$, a single ancilla suffices, when $k\leq 2$, and $k-2$ ancillas suffice, when $k>2$.
\end{theorem}

\begin{proof}
The left-to-right direction is an immediate consequence of the fact that the elements of $\gensk$ have entries in $\DDzk$. For the right-to-left direction, we proceed by induction on $k$. The cases of $k=1,2,3$ follow from \cite[Corollary~5.6]{Amy2020}, \cite[Corollary~5.27]{Amy2020}, and \cite[Theorem~1]{Giles2013a}, respectively. Now suppose that $k>3$, let $U\in \DDunnk{2^m}$, and let $(\phi_k,\ket{\lambda_k}): \DDuz{k}\to \DDuz{k-1}$ be the catalytic embedding of \cref{prop:emb-step}. Then $\phi_k(U)\in \DDunnkk{2^{m+1}}{k-1}$. Thus, by the induction hypothesis, there exists an $(m+1)$-qubit circuit $C$ for $\phi_k(U)$ over $\gensn{2^{k-1}}$ that uses no more than $k-3$ ancillas. For every state $\ket{u}$, we then have
\begin{equation}
\label{eq:exactsynthind}
C(\ket{u}\otimes \ket{\lambda_k}) = \phi_k(U)(\ket{u}\otimes \ket{\lambda_k}) = (U\ket{u})\otimes \ket{\lambda_k}.
\end{equation}
Now let $D$ be the circuit defined by $D=(I\otimes (T_{2^k} H))^\dagger \circ C\circ (I\otimes (T_{2^k} H))$. This is a circuit over $\gensk$, since $X=\hgate S^2 \hgate^\dagger$ implies that $H$ can be expressed as
\[
H = \zeta_8^\dagger \hgate=  X  (T_{2^k}^\dagger)^{2^{k-3}} X  (T_{2^k}^\dagger)^{2^{k-3}} \hgate
\]
when $k\geq 3$. By \cref{eq:exactsynthind}, and since $\ket{\lambda_k} = T_{2^k} H \ket{0}$, we then have
\begin{align*}
D(\ket{u}\otimes \ket{0})
&= (I\otimes (T_{2^k} H))^\dagger \circ C\circ (I\otimes (T_{2^k} H)) (\ket{u}\otimes \ket{0}) \\ 
&= (I\otimes (T_{2^k} H))^\dagger \circ C(\ket{u}\otimes \ket{\lambda_k}) \\ 
&= (I\otimes (T_{2^k} H))^\dagger((U\ket{u})\otimes \ket{\lambda_k}) \\
&= (U\ket{u})\otimes \ket{0}.
\end{align*}
That is, $D$ represents $U$ exactly and uses no more than $k-2$ ancillas, which completes the proof.
\end{proof}

The circuit constructed in the inductive step of \cref{thm:exactsynth} is depicted in \cref{fig:inductivecirc}. The ancillary qubits used by $C$ are not represented in \cref{fig:inductivecirc} (just as they are kept implicit in the proof of the theorem).

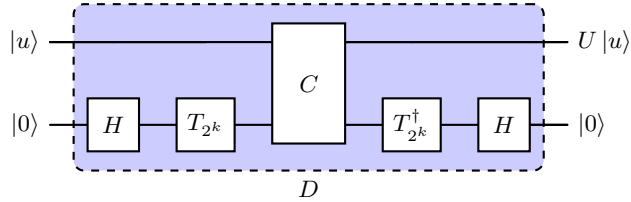
\begin{figure}[t]
\centering
\begin{quantikz}
\lstick{$\ket{u}$} & \qw \gategroup[2,steps=5,style={dashed,rounded corners,fill=blue!20, inner xsep=2pt},background,label style={label position=below,anchor=north,yshift=-0.2cm}]{{$D$}} &\qw & \gate[wires=2]{~~C~~} & \qw & \qw & \qw\rstick{$U\ket{u}$}\\
\lstick{$\ket{0}$} & \gate{H}\vphantom{T_{2^k}^\dagger}\hphantom{~} & \gate{T_{2^k}}\vphantom{T_{2^k}^\dagger} & & \gate{T_{2^k}^\dagger}\vphantom{T_{2^k}^\dagger} & \gate{H}\vphantom{T_{2^k}^\dagger}\hphantom{~} & \qw\rstick{$\ket{0}$}      
\end{quantikz}   
\caption{The circuit constructed in the proof of \cref{thm:exactsynth}.\label{fig:inductivecirc}}
\end{figure}

The construction of \cref{thm:exactsynth} can be used to give an alternative proof of \cite[Corollary~5.27]{Amy2020} and \cite[Theorem~1]{Giles2013a}, albeit one that uses more ancillas than is necessary. In the proof of \cref{thm:exactsynth}, the cases of $k=1$, $k=2$, and $k=3$ are all treated as base cases. Instead, one could use only the case of $k=1$ as the base case and establish the cases of $k=2$ and $k=3$ inductively. The resulting circuit would then use $k$ ancillas to represent an element of $\DDuzk$ for all $k$, rather than $k-2$ ancillas when $k>2$, as in the current proof. 

\section{Conclusion}
\label{sec:conc}

Several questions arise from this work. Firstly, can the proof \cref{thm:exactsynth} be modified so as to produce smaller circuits? The size of the circuits produced by the theorem depends on the exact synthesis algorithm applied in the base case, but the produced circuits are likely to remain large, even if improved synthesis methods such as \cite{improvedtof,ctsynthesis,improvedct,russell} are used. Lowering this cost is an important avenue for future work. Secondly, can \cref{thm:exactsynth} be generalized to Clifford-cyclotomic gate sets of degree $n\neq 2^k$ or can such an extension be shown to be impossible? Preliminary research indicates that arbitrary roots of unity can be represented using circuits over $\s{X, CX, CCX, H\otimes H}$ in the presence of appropriate catalysts, but the construction is more intricate than the one presented here. Finally, and further afield, can \cref{thm:exactsynth} be used to develop algorithms for the approximation of unitaries using Clifford-cyclotomic circuits, following prior work such as \cite[Appendix~A]{lower}, \cite{shorter}, or \cite{gridsynth}?

\subsubsection{Acknowledgements:}

The authors would like to thank Sarah Meng Li, Vadym Kliuchnikov, Kira Scheibelhut, and Peter Selinger for insightful comments on an earlier version of this note. The circuit diagram in this note was typeset using Quantikz \cite{quantikz}. 

\subsubsection{Disclosure of Interests:}

MA was supported by the Canada Research Chairs program. MA and NJR were supported by the Natural Sciences and Engineering Research Council of Canada (NSERC). SK was supported by ONR, whose sponsorship and continuing guidance of the ILIR program has made this research possible. These efforts were funded under ONR award N0001423WX00070. SK, SSM, and WM were supported by Naval Innovative Science and Engineering funding. WM was supported by the U.S. Department of Energy, Office of Science, National Quantum Information Science Research Centers, Quantum Systems Accelerator.

\bibliographystyle{splncs04}
\bibliography{multiqubitcyclo}

\end{document}